\documentclass[11pt]{article}
\usepackage{amsmath}
\usepackage{graphicx,psfrag,epsf}
\usepackage{enumerate}
\usepackage{psfrag,epsf}
\usepackage{url} 
\usepackage{amsmath,amssymb,amsfonts,amsthm,mathtools,mathrsfs}

\usepackage{bm,bbm}
\usepackage{color}
\usepackage{subfigure}
\usepackage{algorithm,algpseudocode}
\usepackage[symbol]{footmisc}
\usepackage{scalefnt}
\usepackage{authblk} 
\usepackage{multirow,centernot}
\usepackage[colorlinks=true, citecolor=blue, urlcolor=blue]{hyperref}
\usepackage{natbib}
\usepackage{dsfont}
\usepackage[title]{appendix}
\usepackage{newtxtext,newtxmath} 
\usepackage{mhchem} 

\input cyracc.def

\usepackage[left=1in,top=1.1in,right=0.5in,bottom=1in]{geometry}

\theoremstyle{definition}

\newtheorem{theorem}{Theorem}[section]

\newtheorem{example}{Example}[section]

\newtheorem{lemma}[theorem]{Lemma}
\newtheorem{proposition}[theorem]{Proposition}
\newtheorem{remark}[theorem]{Remark}

\makeatletter
\def\@seccntformat#1{\@ifundefined{#1@cntformat}%
	{\csname the#1\endcsname\quad}
	{\csname #1@cntformat\endcsname}
}
\makeatother

\newif\ifShowComments
\ShowCommentstrue
\def\strutdepth{\dp\strutbox}
\def\druk#1{\strut\vadjust{\kern-\strutdepth
        {\vtop to \strutdepth{%
                \baselineskip\strutdepth\vss
                        \llap{\hbox{#1}\quad}\null}}}}

\title{\bf
Closed-form estimators for an exponential family derived from likelihood equations
}

\author[]{Roberto Vila}
\author[]{Eduardo Nakano}
\author[]{Helton Saulo  \thanks{Corresponding author: heltonsaulo@gmail.com} }
\affil[]{Department of Statistics, University of
	Bras\'ilia, 70910-900, Bras\'ilia, Brazil}
\begin{document}
	\maketitle 
	\begin{abstract}
		{
In this paper, we derive closed-form estimators for the parameters of some probability distributions belonging to the exponential family. A bootstrap bias-reduced version of these proposed closed-form estimators are also derived. A Monte Carlo simulation is performed for the assessment of the estimators. The results are seen to be quite favorable to the proposed bootstrap bias-reduce estimators.
}

	\end{abstract}
	\smallskip
	\noindent
	{\small {\bfseries Keywords.} {Exponential family $\cdot$ Maximum likelihood method $\cdot$ Monte Carlo simulation $\cdot$ \verb+R+ software.}}
	\\
	{\small{\bfseries Mathematics Subject Classification (2010).} {MSC 60E05 $\cdot$ MSC 62Exx $\cdot$ MSC 62Fxx.}}
	
	\section{Introduction}
	\noindent

The maximum likelihood (ML) method is certainly the most popular technique for obtaining point estimates of the parameters of a distribution. However, it is rare for multi-parameter probability distributions to have closed-form expressions for the ML estimators. Although numerical optimization algorithms are used to solve ML estimation problems that have no closed-form expressions, the availability of a closed-form estimator can avoid some difficulties that can occur when using iterative methods. For example, the problems of estimation in the limits of parametric space can cause the non-convergence of the algorithm and the requirement to adopt a multi-start approach increases computing time. In fact, the low computational cost is the main advantage of closed-form estimators, especially in the field of real-time processing.

In this context, the aim of this work is to provide closed-form estimators for the parameters of a parametric set of probability distributions that belongs to the exponential family \citep{Efron2022} of the following form
\begin{align*}
f(x;\psi)
=
h(x)
\exp\left(\sum_{j=1}^{2}\theta_j(\psi)T_j(x)-\log(c(\theta(\psi)))\right),
\quad 
x>0,\  \psi=(\mu,\sigma), \ \mu,\sigma>0,
\end{align*}
where $\theta(\psi)=(\theta_1(\psi),\theta_2(\psi))$, 
$T_1:(0,\infty)\to (0,\infty)$ is a
real strictly monotone twice differentiable function,
\begin{align*}
	h(x)={\vert T_1'(x)\vert\over T_1(x)}, \quad 
	\theta_1(\psi)=-\mu\sigma, \quad 
	\theta_2(\psi)=\mu,  \quad 
	T_2(x)=\log(T_1(x)),  
\end{align*} 
%
%
with
%
$
c(\theta(\psi))
=
{\Gamma(\mu)/(\sigma^\mu\mu^\mu)}
$
being the normalizing constant.
%
In simple terms, $f(x;\psi)$ has the following simple form
\begin{align}\label{pdf-1}
f(x;\psi)
=
{(\mu\sigma)^\mu \over \Gamma(\mu)}\,
{\vert T_1'(x)\vert\over T_1(x)}\,
\exp\left\{-\mu \sigma T_1(x)+\mu\log(T_1(x))\right\},
\quad 
x>0,\  \psi=(\mu,\sigma), \ \mu,\sigma>0.
\end{align}
In the above, $T_1'(x)$ denotes the derivative of $T_1(x)$ with respect to $x$.
Note that a family of distributions equivalent to \eqref{pdf-1} has appeared in \cite{Nascimento2014}.

Table \ref{table:1} presents some examples of generators $T_1(x)$ for use in \eqref{pdf-1}.
\begin{table}[H]
	\caption{Some examples of 
		generators $T_1(x)$ of exponential family \eqref{pdf-1}.}
	\vspace*{0.15cm}
	\centering 
	\renewcommand{\arraystretch}{-0.8}
		\resizebox{\linewidth}{!}{
	\begin{tabular}{lcccrll} 
		\hline
		Distribution & $\mu$ & $\sigma$ & $T_1(x)$ & Parameters
		\\ [0.5ex] 
		\noalign{\hrule height 1.0pt}
		\\ 
		Nakagami \citep{Laurenson1994}
		&  $m$  & ${1\over \Omega}$   & $x^2$ &  $m\geqslant {1\over 2}$, \
		$\Omega>0$ 
		\\ [2.0ex]	
		Maxwell-Boltzmann \citep{Dunbar1982}
		&  ${3\over 2}$  & ${1\over 3\beta^2}$   & $x^2$ &  $\beta>0$ 
		\\ [2.0ex]	
		Rayleigh \citep{Rayleigh1880}
		&  $1$  & ${1\over 2\beta^2}$   & $x^2$ &  $\beta>0$ 
		\\ [2.0ex]		
		Gamma \citep{Stacy1962}
		&  $\alpha$  &  ${1\over\alpha\beta}$ & $x$ &  $\alpha, \beta > 0$ 
		\\ [2.0ex]		
Inverse gamma \citep{Cook2008}
&  $\alpha$  &  ${1\over\alpha\beta}$ & ${1\over x}$ &  $\alpha, \beta > 0$ 
		\\ [2.0ex]
 		$\delta$-gamma \citep{Rahman2014}
		&  ${\beta\over \delta}$  &  ${1\over \beta}$ & $x^\delta$ & 
		$\delta, \beta > 0$ 
		\\ [2.0ex]
		Weibull \citep{Johnson1994}		
		&  $1$ & ${1\over \beta^\delta}$ & $x^\delta$ &  
		$\delta, \beta>0$ 
		\\ [2.0ex]		
		Inverse Weibull (Fréchet) \citep{khan2008}
		& $1$ & ${1\over \beta^\delta}$ & ${1\over x^\delta}$ &  
		$\delta, \beta>0$ 
		\\ [2.0ex] 
		Generalized gamma \citep{Stacy1962}
		& ${\alpha\over \delta}$ & ${\delta\over \alpha \beta^\delta}$ & $x^\delta$ &  $\alpha,\delta, \beta>0$ 
		\\ [2.0ex] 
		Generalized inverse gamma \citep{Lee1991}
		&  ${\alpha\over \delta}$ &  ${\delta\over \alpha\beta^{\delta}}$ & ${1\over x^\delta}$ &  
		$\alpha,\delta, \beta>0$ 
		\\ [2.0ex] 
		New log-generalized gamma$^{*}$
		& ${\alpha\over \delta}$ & ${\delta\over \alpha \beta^\delta}$ & $[\exp(x)-1]^\delta $ & $\alpha,\delta>0$,  $\beta>0$
		\\ [2.0ex] 
New log-generalized inverse gamma$^{*}$
& ${\alpha\over \delta}$ & ${\delta\over \alpha \beta^\delta}$ & $\big[\exp\big({1\over x}\big)-1\big]^\delta $ & $\alpha,\delta>0$,  $\beta>0$			
				\\ [2.0ex] 
		New exponentiated generalized gamma$^{*}$
		& ${\alpha\over \delta}$ & ${\delta\over \alpha \beta^\delta}$ & $\log^\delta(x+1) $ & $\alpha,\delta>0$,  $\beta>0$
					\\ [2.0ex] 
	New exponentiated generalized inverse gamma$^{*}$
	& ${\alpha\over \delta}$ & ${\delta\over \alpha \beta^\delta}$ & $\log^\delta\big({1\over x}+1\big) $ & $\alpha,\delta>0$,  $\beta>0$
						\\ [2.0ex] 
	New modified log-generalized gamma$^{*}$
	& ${\alpha\over \delta}$ & ${\delta\over \alpha \beta^\delta}$ & $\exp^\delta\big(x-{1\over x}\big) $ & $\alpha,\delta>0$,  $\beta>0$
							\\ [2.0ex] 	
	New extended log-generalized gamma$^{*}$
		& ${\alpha\over \delta}$ & ${\delta\over \alpha \beta^\delta}$ & $x^\delta[\exp(x)-1]^\delta$ & $\alpha,\delta>0$,  $\beta>0$
		\\ [2.0ex] 
		Chi-squared \citep{Johnson1994}	
		& ${\nu\over 2}$ & ${1\over \nu}$ & $x$ & $\nu>0$ 
		\\ [2.0ex] 	
		Scaled inverse chi-squared  \citep{Bernardo1993}
		& ${\nu\over 2}$ & $\tau^2$ &  ${1\over x}$ &  $\nu, \tau^2>0$ 
		\\ [2.0ex] 	
		Gompertz \citep{Gompertz1825}
		& $1$ & $\alpha$ &  $\exp(\delta x)-1$ &  $\alpha, \delta>0$ 
		\\ [2.0ex] 		
		Modified Weibull extension \citep{Xie2022}
		& $\lambda\alpha$ & $1$ &  $\exp\big[ \left({x\over\alpha}\right)^\beta\big]-1$ &  $\alpha,\lambda,\beta>0$
		\\ [2.0ex] 		
		Traditional Weibull \citep{Nadarajah2005}
		& $1$ & $a$ &  $x^b[\exp( cx^d)-1]$ &  $a, d>0$, \ $b,c\geqslant 0$
		\\ [2.0ex] 		
		Flexible Weibull \citep{Bebbington2007}
		& $1$ & $a$ &  $\exp\left(b x-{c\over x}\right)$ &  $a, b, c>0$
		\\ [2.0ex] 		
		Burr type XII (Singh-Maddala) \citep{Burr1942}
		& $1$ & $k$ &  $\log(x^c+1)$ &  $c, k>0$
		\\ [2.0ex] 		
		Dagum (Mielke Beta-Kappa) \citep{Dagum1975}
		& $1$ & $k$ &  $\log\big({1\over x^c}+1\big)$ &  $c, k>0$
		\\ [1.5ex] 	
		\hline	
	\end{tabular}
	}
	\label{table:1} 
\end{table}
To the best of our knowledge, the distributions in Table \ref{table:1} highlighted with ``*'' are new in the literature.
{\color{black} 
Since many of them are derived from generalized gamma-type models, they naturally inherit potential applications in areas where the gamma, Weibull, and exponential distributions are commonly employed, including reliability analysis, survival studies, hydrology, actuarial science, and income modeling. In this sense, the practical relevance of these new distributions is supported by the extensive literature and applications associated with their parent distributions. A detailed empirical assessment of these new models is beyond the scope of the present work and represents an interesting direction for future research.
}


%

Several studies have presented closed-form estimators derived from likelihood equations. \cite{YCh2016} derived a closed-form estimator of the parameters of the gamma distribution from generalised gamma distribution. The generalized gamma can be obtained by a power transformation $Y=X^{1/\gamma}$, where $X$ has gamma distribution and $\gamma>0$ is a power parameter. Similar ideas were adopted to develop closed-form estimators for the parameter of the Nakagami distribution \citep{RLR2016, Zhao2021} and the weighted Lindley distribution \citep{Kim2020}. Recently, \cite{Kim2022} presented a new procedure for obtaining a closed-form estimator for the family of distributions using an extension of the Box-Cox transformation. This work adopts the same idea of \cite{YCh2016} and \cite{Cheng-Beaulieu2002} to develop a closed-form estimator for the parameters of some probability distributions of the exponential family \eqref{pdf-1}.

The rest of the paper proceeds as follows. In Section~\ref{sec:2}, we describe briefly a generalization of exponential family and its basic properties. The new proposed estimation methods and some
asymptotic results are described in Sections~\ref{The New Estimators} and \ref{largeprop}, respectively. Finally, in Section~\ref{sec:simulation}, we perform a Monte Carlo simulation study to assess the performance of a bootstrap bias-reduced version of these proposed closed-form estimators.



\section{A generalization of exponential family}\label{sec:2}

{\color{black}
	It is worth noting that the simple power transformation $Y=X^{1/p}$, $p>0$, will be used repeatedly throughout the paper to generate additional members of the family introduced in Proposition~\ref{Stochastic representation}. Although this transformation is well known, its combination with Proposition~\ref{Stochastic representation} provides a convenient way to obtain new distributions within the proposed framework. As will be seen in Section~\ref{The New Estimators}, many of these models also admit simple likelihood equations and closed-form estimators.
}

By using the increasing transformation $Y=X^{1/p}$, $p>0$, considered in \cite{YCh2016} and \cite{Cheng-Beaulieu2002}, where $X$ has the probability density function (PDF) in \eqref{pdf-1}, it is clear that the PDF of $Y$ {\color{black} can be} written as
\begin{align}\label{dist-gen-exp}
f(y;\psi,p)
=
p\, {(\mu\sigma)^\mu \over \Gamma(\mu)}\,
{\vert T_1'(y^p)\vert y^{p-1}\over T_1(y^p)}\,
\exp\left\{-\mu \sigma T_1(y^p)+\mu\log(T_1(y^p))\right\},
\quad 
y>0,
\end{align}
where $\psi=(\mu,\sigma)$ and $\mu,\sigma, p>0$.

\subsection{Stochastic representation}

{\color{black}
A random variable $X$ is said to have a Gamma distribution with shape
parameter $\alpha>0$ and rate parameter $\lambda>0$, denoted by
$X\sim{\rm Gamma}(\alpha,\lambda)$, if its PDF is given by
\[
f_X(x)
=
\frac{\lambda^\alpha}{\Gamma(\alpha)} \, 
x^{\alpha-1}\exp(-\lambda x),
\quad x>0.
\]
In this parametrization,
\[
\mathbb E(X)=\frac{\alpha}{\lambda},
\quad
\operatorname{Var}(X)=\frac{\alpha}{\lambda^2}.
\]
}

\begin{proposition}\label{Stochastic representation}
	If $Z\sim {\rm Gamma}(\mu,{\color{black}\mu\sigma})$ then $Y$ has the stochastic representation 
	$
	Y\stackrel{d}{=}[T^{-1}_1(Z)]^{1/p},
	$
	with $\stackrel{d}{=}$ being equality in distribution and $T^{-1}_1$ denoting the inverse function of $T_1$.
	
	Conversely, if $Y$ is distributed according \eqref{dist-gen-exp}, then
	$
	Z\stackrel{d}{=} T_1(Y^p)\sim {\rm Gamma}(\mu,{\color{black}\mu\sigma})
	$
\end{proposition}
\begin{proof}
Let $Z\sim {\rm Gamma}(\mu,{\color{black}\mu\sigma})$. A simple observation shows that (for $y>0$)
	\begin{align}\label{ident-moments}
	\mathbb{P}([T^{-1}_1(Z)]^{1/p}\leqslant y)
	=
	\mathbb{P}(T^{-1}_1(Z)\leqslant y^p)
	=
	\begin{cases}
	\mathbb{P}(Z\leqslant T_1(y^p)), & \text{if} \ T_1 \ \text{is increasing},
	\\[0,2cm]
	1-\mathbb{P}(Z\leqslant T_1(y^p)), & \text{if} \ T_1 \ \text{is decreasing}.
	\end{cases}
	\end{align}
	Taking the derivative with respect to $y$ both sides, the PDF of $[T^{-1}_1(Z)]^{1/p}$ is written as
	\begin{align*}
	f_{[T^{-1}_1(Z)]^{1/p}}(y)
	&=
	py^{p-1} \vert T'_1(y^p)\vert
	f_Z(y^p;\mu,1/(\mu\sigma)),
	\quad 
	Z\sim {\rm Gamma}\left(\mu, {\color{black}\mu\sigma}\right)
	\\[0,2cm]
	&=
	f(y;\psi,p).
	\end{align*}
	This completes the proof.
\end{proof}

\subsection{Quantiles}

Given $p\in(0,1)$, denote by $Q_Y(p)$ the $p$-quantile of a random variable $Y$ distributed according \eqref{dist-gen-exp}.

In what follows we find simple expressions for $Q_Y(p)$. Indeed,
from Proposition \ref{Stochastic representation} and from \eqref{ident-moments}, we have
\begin{align*}
p=
\mathbb{P}(Y\leqslant Q_Y(p))
	=
\begin{cases}
\mathbb{P}(Z\leqslant T_1(Q_Y^p(p))), & \text{if} \ T_1 \ \text{is increasing and} \ Z\sim {\rm Gamma}(\mu, {\color{black}\mu\sigma}),
\\[0,2cm]
\mathbb{P}\big({Z^*}\leqslant {1\over T_1(Q_Y^p(p))}\big), & \text{if} \ T_1 \ \text{is decreasing and} \ Z^*={1\over Z}\sim {\rm Inv}$-${\rm Gamma}(\mu, {\color{black}\mu\sigma}).
\end{cases}
\end{align*}
This shows that
\begin{align*}
Q_Z(p)=T_1(Q_Y^p(p)) \quad  \Longleftrightarrow \quad Q_Y(p)=[T_1^{-1}(Q_Z(p))]^{1/p},
\end{align*}
whenever $T_1$ is increasing, and
\begin{align*}
Q_{Z^*}(p)={1\over T_1(Q_Y^p(p))} \quad \Longleftrightarrow \quad Q_Y(p)=\left[T_1^{-1}\left({1\over Q_{Z^*}(p)}\right)\right]^{1/p},
\end{align*}
whenever $T_1$ is decreasing.

\subsection{Moments}

\begin{proposition}\label{proof-teo-moments}
	If $Y$ is distributed according \eqref{dist-gen-exp}, then the positive order moments of $Y$, denoted by $\mathbb{E}(Y^q)$, $0<q<ps,$ exist
	whenever $T_1(x)\geqslant Cx^{s}$, $x>0$, $T_1$  increasing, for some $C>0$ and $s> 0$.
\end{proposition}
\begin{proof}
By using the well-known formula
\begin{align*}
	\mathbb{E}(Y^q)=q\int_0^\infty y^{q-1}\mathbb{P}(Y>y){\rm d}y, \quad Y>0, \ q>0,
\end{align*}
for $Y$ distributed according \eqref{dist-gen-exp},
we have (for some $0<a<\infty$)
\begin{align}
\mathbb{E}(Y^q)
&=q\int_0^a y^{q-1}\mathbb{P}(Y>y){\rm d}y
+
q\int_a^\infty y^{q-1}\mathbb{P}(Y>y){\rm d}y
\nonumber
\\[0,2cm]
&\leqslant
a^q
+
q\int_a^\infty y^{q-1}\mathbb{P}(Y>y){\rm d}y
\nonumber
\\[0,2cm]
&=
a^q
+
q\int_a^\infty y^{q-1} \mathbb{P}(Z> T_1(y^p)) {\rm d}y,
\quad Z\sim {\rm Gamma}\left(\mu, {1\over \mu\sigma}\right),
\label{ide-ep}
\end{align}
where in the last line we have used \eqref{ident-moments}. By Markov's inequality, the expression in \eqref{ide-ep} is at most
\begin{align*}
	a^q
	+
	q\mathbb{E}(Z)
	\int_a^\infty {y^{q-1}\over T_1(y^p)} {\rm d}y
	\leqslant
	a^q
	+
	{q\mathbb{E}(Z)\over C}
	\int_a^\infty {1\over y^{ps-q+1}} {\rm d}y,
\end{align*}
where in the last inequality we used the fact that $T_1(x)\geqslant Cx^{s}$, $x>0$.
For $ps>q$ the last integral converges. Hence the existence of the positive order moments of $Y$ follows.
\end{proof}

\begin{remark}
Note that the generator $T_1(x)=x^b[\exp( cx^d)-1]$, $x>0$, of traditional Weibull distribution in Table \ref{table:1} satisfies the condition $T_1(x)\geqslant Cx^{s}$ of Proposition \ref{proof-teo-moments} with $C=c$ and $s=b+d$. In particular, the generators of Gompertz  and modified Weibull extension in Table \ref{table:1} also satisfy this condition.
\end{remark}

\begin{proposition}\label{dist-exp}
	If $Y$ is distributed according \eqref{dist-gen-exp} with
	$T_1(x)=\log(x^s+1)$, $x>0$, for some $s\neq 0$, then the real moments of $Y$, denoted by $\mathbb{E}(Y^q)$, $q<\min\{0,ps\}$, exist.
\end{proposition}
\begin{proof}
As a by-product from the proof of Proposition \ref{proof-teo-moments}, it is sufficient to prove that	
\begin{align*}
	I=\int_a^\infty {y^{q-1}\over T_1(y^p)} {\rm d}y<\infty,
\end{align*}
for some $0<a<\infty$. Indeed, by using the inequality $\log(x)>1-1/x$ we have $T_1(x)\geqslant x^s/(x^s+1)$. Hence
\begin{align*}
	I\leqslant \int_a^\infty {1\over y^{-q+1}} {\rm d}y + \int_a^\infty {1\over y^{ps-q+1}} {\rm d}y<\infty,
\end{align*}
provided $q<\min\{0,ps\}$.
\end{proof}

{\color{black}
\begin{remark}
	Propositions~\ref{proof-teo-moments} and~\ref{dist-exp}
	illustrate how the form of the transformation function $T_1$
	affects the tail behaviour of the resulting distribution and,
	consequently, the existence of moments. Proposition~\ref{proof-teo-moments}
	provides a simple sufficient condition for the existence of positive moments
	for a broad class of models, whereas Proposition~\ref{dist-exp}
	shows that, for the specific choice $T_1(x)=\log(x^s+1)$,
	the range of existing moments can be characterized explicitly.
\end{remark}
}

\begin{remark}
	Note that Proposition \ref{dist-exp} includes the generators  $T_1(x)$ of Burr Type XII  and Dagum distributions (see Table \ref{table:1}).
\end{remark}

\begin{proposition}\label{dist-exp-1}
	If $Y$ is distributed according \eqref{dist-gen-exp} with $T_1(x)=C x^{-s}$, $x>0$, for some $C>0$ and $s\neq 0$, then the real moments of $Y$ are  given by
	\begin{align*}
	\mathbb{E}(Y^q)
	=
	\left({\mu\sigma\over C}\right)^{q\over ps} \,
	\dfrac{\Gamma\left(\mu-{q\over ps}\right)}{\Gamma(\mu)},
	\quad \text{where} \ \mu-{q\over ps}>0.
	\end{align*}
\end{proposition}
\begin{proof}
By using Proposition \ref{Stochastic representation}, we have	
$
	\mathbb{E}(Y^q)
	=
	\mathbb{E}\{[T^{-1}_1(Z)]^{q/p}\}
	=
	C^{-q/(ps)}
	\mathbb{E}[Z^{-q/(ps)}]
$,
for
$Z\sim {\rm Gamma}(\mu,{\color{black}\mu\sigma})$.
Since $\mathbb{E}(X^\nu)={\color{black}\lambda^{-\nu}} \Gamma(\nu+\alpha)/\Gamma(\alpha)$ for $Z\sim {\rm Gamma}(\alpha,{\color{black}\lambda})$ and $\nu>-\alpha$, the proof follows.
\end{proof}

\begin{remark}
	Note that Proposition \ref{dist-exp-1} includes the generators  $T_1(x)$ of Nakagami, Maxwell-Boltzmann, Rayleigh, gamma, inverse gamma, $\delta$-gamma, Weibull, inverse Weibull, generalized gamma, generalized inverse gamma, chi-squared and scaled inverse chi-squared (see Table \ref{table:1}).
\end{remark}

\section{The new estimators}\label{The New Estimators}

Let $\{Y_i : i = 1,\ldots , n\}$ be a univariate random sample of size $n$ from  $Y$, {\color{black} where $Y$ has PDF given in \eqref{dist-gen-exp}.}

%
The (random) likelihood function for $(\psi,p)$ is written as
\begin{align*}
L(\psi,p)    
=
p^n\, {(\mu\sigma)^{n\mu}  \over \Gamma^n(\mu)}\,
\prod_{i=1}^{n}
{\vert T_1'(Y_i^p)\vert Y_i^{p-1}\over T_1(Y_i^p)}\,
\exp\left\{-\mu \sigma \sum_{i=1}^{n} T_1(Y^p_i)+\mu\sum_{i=1}^{n}\log(T_1(Y^p_i))\right\}.  
\end{align*}
Consequently, the (random) log-likelihood function for $(\psi,p)$ is written as
\begin{align*}
\log(L(\psi,p))    
&=
n\log(p)
+
{n\mu}\log(\mu) 
+ 
{{n\mu} \log(\sigma)
-
n\log(\Gamma(\mu))}
\\[0,2cm]
&
+
\sum_{i=1}^{n}
\log(\vert T_1'(Y_i^p)\vert)
+
(p-1)
\sum_{i=1}^{n}
\log(Y_i)
-
\mu \sigma \sum_{i=1}^{n} T_1(Y^p_i)
+
(\mu-1)\sum_{i=1}^{n}\log(T_1(Y^p_i)).  
\end{align*}

A simple calculus shows that the elements of the (random) score vector are given by
\begin{align*}
 {\partial \log(L(\psi,p)) \over\partial \mu}
 &=
n\log(\mu)+n\log(\sigma)+n-n\psi^{(0)}(\mu)
-
\sigma \sum_{i=1}^{n} T_1(Y^p_i)
+
\sum_{i=1}^{n}\log(T_1(Y^p_i)),  
 \\[0,2cm]
  {\partial \log(L(\psi,p)) \over\partial \sigma}
 &=
  {n\mu \over \sigma} 
  -
  \mu \sum_{i=1}^{n} T_1(Y^p_i),
 \\[0,2cm]
  {\partial  \log(L(\psi,p)) \over\partial p}
 &=
 {n\over p}
 +
 \sum_{i=1}^{n}
{T_1''(Y^p_i)\over T_1'(Y_i^p)}\, Y^p_i \log(Y_i)
+
\sum_{i=1}^{n}
\log(Y_i)
\\[0,2cm]
&
-
\mu \sigma \sum_{i=1}^{n} T_1'(Y^p_i) Y^p_i \log(Y_i)
+
(\mu-1)\sum_{i=1}^{n} {T_1'(Y^p_i)\over T_1(Y^p_i)}\,  Y^p_i \log(Y_i), 
\end{align*}
where $\psi^{m}(x)=\partial^{m+1} \log(\Gamma(x))/\partial x^{m+1}$ is the polygamma function of order $m$. Setting these equal to zero and solving the system of equations gives the ML estimators of $(\psi, p)$.
In particular, by resolving  ${\partial \log(L(\psi,p))/\partial \sigma}=0$, we can express $\sigma$ as a function of $p$:
%
\begin{align}\label{mle-sigma}
{\sigma}(p)
=
{n\over\sum_{i=1}^{n} T_1(Y_i^{{p}})}.
\end{align}
Substitute \eqref{mle-sigma} into ${\partial \log(L(\psi,p))/\partial p}=0$ to give
%
\begin{align}\label{mle-mu}
	{\mu}(p)
	=
		 \dfrac{
		 	{n\over p}
		 	+ 
		 	\sum_{i=1}^{n}
		 	\log(Y_i)
		 	+ 
		 	\sum_{i=1}^{n}
		 	\left[{T_1''(Y^p_i)\over T_1'(Y_i^p)}-{T_1'(Y^p_i)\over T_1(Y^p_i)}\right] Y^p_i \log(Y_i)
	 	}
	 	{{\sigma}(p) \sum_{i=1}^{n} T_1'(Y^p_i) Y^p_i \log(Y_i)
		 	-
		 \sum_{i=1}^{n} {T_1'(Y^p_i)\over T_1(Y^p_i)}\,  Y^p_i \log(Y_i)}.
\end{align}
Furthermore, notice that 
${p}$ is obtained solving the non-linear equation
\begin{align}\label{mle-p}
	\log({\mu}(p))-\psi^{(0)}({\mu}(p))
	=
	\log\left({1\over\sigma(p)}\right)
	-
		{1\over n}
	\sum_{i=1}^{n}\log(T_1(Y^{{p}}_i)).
\end{align}

Now, return to the distribution in \eqref{pdf-1}. We take $p=1$ in \eqref{mle-sigma} and \eqref{mle-mu} to obtain the new estimators for $\mu$ and $\sigma$ as
%
\begin{align}\label{mle-sigma-1}
\widehat{\sigma}
=
{1\over{1\over n} \sum_{i=1}^{n} T_1(Y_i)}
\end{align}
and
\begin{align}\label{mle-mu-1}
\widehat{\mu}
=
\dfrac{{1\over n}\sum_{i=1}^{n} T_1(Y_i) 
	\left\{
	1+ 
	{1\over n}
		\sum_{i=1}^{n}
	\log(Y_i)
	+ 
		{1\over n}
	\sum_{i=1}^{n}
	\left[{T_1''(Y_i)\over T_1'(Y_i)}-{T_1'(Y_i)\over T_1(Y_i)}\right] Y_i \log(Y_i)
	\right\}}{{1\over n} \sum_{i=1}^{n} T_1'(Y_i) Y_i \log(Y_i)
	-
	{1\over n}\sum_{i=1}^{n} T_1(Y_i) \ {1\over n}\sum_{i=1}^{n} {T_1'(Y_i)\over T_1(Y_i)}\,  Y_i \log(Y_i)}.
\end{align}
%

Note that the ML estimator of $\sigma$, denoted by $\widehat{\sigma}_{\rm MLE}$, satisfies the identity $\widehat{\sigma}_{\rm MLE}=\widehat{\sigma}$. Furthermore,
the ML estimator of $\mu$, denoted by $\widehat{\mu}_{\rm MLE}$, is solution of the following equation, which is obtained from  \eqref{mle-p} by taking $p=1$, 
\begin{align}\label{mle-mu-p=1}
\log({\mu})-\psi^{(0)}({\mu})
=
\log\left({1\over \sigma}\right)
-
{1\over n}
\sum_{i=1}^{n}\log(T_1(Y_i)).
\end{align}

\begin{proposition}
The equation \eqref{mle-mu-p=1} is an unbiased estimating equation for $\mu$, i.e.,
\begin{align*}
\log({\mu})-\psi^{(0)}({\mu})
=
\mathbb{E}
\left[
\log\left({1\over \sigma}\right)
-
{1\over n}
\sum_{i=1}^{n}\log(T_1(Y_i))
\right].
\end{align*}
\end{proposition}
\begin{proof}
The proof is immediate since
$T_1(Y_i)\sim {\rm Gamma}(\mu,{\color{black}\mu\sigma})$ (Proposition \ref{Stochastic representation}) 
and
$\mathbb{E}[\log(T_1(Y_i))]=\psi^{(0)}(\mu)-\log(\mu\sigma)$, for $i=1,\ldots,n$.
\end{proof}

\begin{proposition}\label{uniqueness}
	The solution in $\mu$ to equation \eqref{mle-mu-p=1} exists and is unique.
\end{proposition}
\begin{proof}
Combining \eqref{mle-mu-p=1} with \eqref{mle-sigma-1}, equation \eqref{mle-mu-p=1} is written as
\begin{align}\label{eq-ess}
\log({\mu})-\psi^{(0)}({\mu})
=
\log\left({1\over n}
\sum_{i=1}^{n}T_1(Y_i)\right)
-
{1\over n}
\sum_{i=1}^{n}\log(T_1(Y_i)).
\end{align}
We denote the right-hand side of \eqref{eq-ess} by $H(Y_1,\ldots,Y_n)$.
By using the Binet's second integral for the gamma function \citep{Mathar2023}:
\begin{align}\label{Binet's second integral}
\log(x)-\psi^{(0)}(x)
=
{1\over 2x}+2\int_0^\infty {t\over (t^2+x^2)[\exp(2\pi t)-1]}{\rm d}t,
\end{align}
note that $\log({\mu})-\psi^{(0)}({\mu})$ is a positive monotone decreasing function of $\mu$ with a finite limit ($=0$) as $\mu\to\infty$.  Since $H(Y_1,\ldots,Y_n)$ is non-negative (by Jensen's inequality) and is constant in $\mu$,
it would then follow that the plots of $\log({\mu})-\psi^{(0)}({\mu})$
and $H(Y_1,\ldots,Y_n)$ would intersect exactly once, at the ML estimator of $\mu$.
\end{proof}

\begin{theorem}\label{theo-compl}
	Let $T_1(x)=C x^{-s}$, for some $C>0$ and $s\neq 0$.
With probability 1, $\widehat{\mu}<2 \widehat{\mu}_{\rm MLE}$.
\end{theorem}
\begin{proof}
If 
$T_1(x)=C x^{-s}$, for some $C>0$ and $s\neq 0$, then 
$T_1^{-1}(x)=(x/C)^{-1/s}$,
$T_1'(x)=-Cs x^{-s-1}$,
$T_1''(x)=Cs(s+1) x^{-s-2}$,
$T_1''(T_1^{-1}(x))=Cs(s+1) (x/C)^{1+2/s}$ 
and
$T_1'(T_1^{-1}(x))=-Cs (x/C)^{1+1/s}$.
Consequently, $1/\widehat{\mu}$, with $\widehat{\mu}$ as given in \eqref{mle-mu-1}, is written as
\begin{align}\label{u-inverse}
    {1\over \widehat{\mu}}
	=
		{
	{1\over n} \sum_{i=1}^{n} Y_i^{-s} \log(Y_i^{-s})
	\over 
		{1\over n}\sum_{i=1}^{n} Y_i^{-s}	
}
	-
	{1\over n}\sum_{i=1}^{n} \log(Y_i^{-s})
.
\end{align}

As $\widehat{\mu}_{\rm MLE}$ is the only solution of equation \eqref{eq-ess} (Proposition \ref{uniqueness}) and from \eqref{Binet's second integral}: $\log(x)-\psi^{(0)}(x)\geqslant (2x)^{-1}$ for $x>0$, we get
\begin{align}\label{u-inverse-1}
	\log\left({1\over n}
\sum_{i=1}^{n}Y_i^{-s}\right)
-
{1\over n}
\sum_{i=1}^{n}\log(Y_i^{-s})
	\geqslant
	{1\over 2 \widehat{\mu}_{\rm MLE}}.
\end{align}

 If we prove that
\begin{align*}
{
	{1\over n} \sum_{i=1}^{n} Y_i^{-s} \log(Y_i^{-s})
	\over 
	{1\over n}\sum_{i=1}^{n} Y_i^{-s}	
}
\geqslant
	\log\left({1\over n}
\sum_{i=1}^{n}Y_i^{-s}\right),
\end{align*}
the proof of theorem follows by combining inequalities \eqref{u-inverse} and \eqref{u-inverse-1}. But this is immediate, since the convexity of the function  $f(x)=x\log(x)$ provides
\begin{align*}
{1\over n}
\sum_{i=1}^{n}Y_i^{-s}
		\log\left({1\over n}
	\sum_{i=1}^{n}Y_i^{-s}\right)
	=
	f\left({1\over n}
	\sum_{i=1}^{n}Y_i^{-s}\right)
	\leqslant
	{1\over n}
	\sum_{i=1}^{n}f(Y_i^{-s})
	=
		{1\over n} \sum_{i=1}^{n} Y_i^{-s} \log(Y_i^{-s}).
\end{align*}
\end{proof}

{\color{black}
\begin{remark}
	Theorem~\ref{theo-compl} provides a simple theoretical comparison between the proposed estimator and the maximum likelihood estimator. In particular, it guarantees that the proposed estimator cannot exceed twice the MLE, showing that both estimators remain on the same scale. The result is included mainly for completeness and to further characterize the relationship between the two estimators.
\end{remark}
}

\section{Large sample properties}\label{largeprop}

In this section, we show that, depending on the {\color{black} form} of $T_1$, the new estimators given in Section \ref{The New Estimators} are strongly consistent and asymptotically Gaussian.

Let $\{Y_i : i = 1,\ldots , n\}$ be a random sample from $Y$ with PDF given in \eqref{pdf-1}.
By applying strong law of large numbers and by using Proposition \ref{Stochastic representation}, we have
\begin{align*}
&{1\over n}\sum_{i=1}^{n} T_1(Y_i) 
\stackrel{\rm a.s.}{\longrightarrow}
\mathbb{E}[T_1(Y)]
=
\mathbb{E}(Z)
=
{1\over \sigma},
\quad 
Z\sim {\rm Gamma}\left(\mu, {\color{black} \mu\sigma}\right);
\\[0,2cm]
&	{1\over n}
\sum_{i=1}^{n}
\log(Y_i)
\stackrel{\rm a.s.}{\longrightarrow}
\mathbb{E}[\log(Y)];
\\[0,2cm]
&{1\over n}
\sum_{i=1}^{n}
\left[{T_1''(Y_i)\over T_1'(Y_i)}-{T_1'(Y_i)\over T_1(Y_i)}\right] Y_i \log(Y_i)
\stackrel{\rm a.s.}{\longrightarrow}
\mathbb{E}\left\{\left[{T_1''(Y)\over T_1'(Y)}-{T_1'(Y)\over T_1(Y)}\right] Y \log(Y)\right\};
\\[0,2cm]
&{1\over n} \sum_{i=1}^{n} T_1'(Y_i) Y_i \log(Y_i)
\stackrel{\rm a.s.}{\longrightarrow}
\mathbb{E}[T_1'(Y) Y \log(Y)];
\\[0,2cm]
&{1\over n}\sum_{i=1}^{n} {T_1'(Y_i)\over T_1(Y_i)}\,  Y_i \log(Y_i)
\stackrel{\rm a.s.}{\longrightarrow}
\mathbb{E}\left[{T_1'(Y)\over T_1(Y)}\,  Y \log(Y)\right].
\end{align*}
A simple application of continuous-mapping theorem \citep{Billingsley1969}  gives
\begin{align}\label{id-1}
\widehat{\mu}
&\stackrel{\eqref{mle-mu-1}}{=}
\dfrac{{1\over n}\sum_{i=1}^{n} T_1(Y_i) 
	\left\{
	1+ 
	{1\over n}
	\sum_{i=1}^{n}
	\log(Y_i)
	+ 
	{1\over n}
	\sum_{i=1}^{n}
	\left[{T_1''(Y_i)\over T_1'(Y_i)}-{T_1'(Y_i)\over T_1(Y_i)}\right] Y_i \log(Y_i)
	\right\}}{{1\over n} \sum_{i=1}^{n} T_1'(Y_i) Y_i \log(Y_i)
	-
	{1\over n}\sum_{i=1}^{n} T_1(Y_i) \ {1\over n}\sum_{i=1}^{n} {T_1'(Y_i)\over T_1(Y_i)}\,  Y_i \log(Y_i)}
\nonumber
\\[0,2cm]
&\stackrel{\rm a.s.}{\longrightarrow}
\dfrac{\mathbb{E}[T_1(Y)] \left\{1+\mathbb{E}[\log(Y)]+\mathbb{E}\left\{\left[{T_1''(Y)\over T_1'(Y)}-{T_1'(Y)\over T_1(Y)}\right] Y \log(Y)\right\}\right\}}{\mathbb{E}[T_1'(Y) Y \log(Y)]-\mathbb{E}[T_1(Y)]\mathbb{E}\left[{T_1'(Y)\over T_1(Y)}\,  Y \log(Y)\right]}
\end{align}
and
\begin{align}\label{id-2}
\widehat{\sigma}
\stackrel{\eqref{mle-sigma-1}}{=}
{1\over{1\over n} \sum_{i=1}^{n} T_1(Y_i)}
\stackrel{\rm a.s.}{\longrightarrow}
\dfrac{1}{\mathbb{E}[T_1(Y)]}
=\sigma.
\end{align}
That is, regardless of the {\color{black} form} of generator $T_1$, the  estimator $\widehat{\sigma}$ for $\sigma$ is strongly consistent. 

{\color{black} Although expression \eqref{id-1} is somewhat involved, it provides a characterization of the almost sure limit of $\widehat{\mu}$. Consequently, it can be used to determine whether the estimator is consistent for a given transformation function $T_1$.}

{
\color{black}
\smallskip
As $T_1(Y)\stackrel{d}{=}Z\sim {\rm Gamma}(\mu, {\color{black}\mu\sigma})$ (see Proposition \ref{Stochastic representation}) and $\mathbb{E}[T_1(Y)]=\mathbb{E}(Z)={1/\sigma}$, from \eqref{id-1}, we get 
\begin{align}\label{id-3}
\widehat{\mu}
\stackrel{\rm a.s.}{\longrightarrow}
\dfrac{
	\mathbb{E}\left\{ 
	1+[1+U_{T_1}(Z)]\log(T_1^{-1}(Z)) \right\}
}
{
	\mathbb{E}\left[\left(\sigma-{1\over Z}\right) T_1'(T_1^{-1}(Z)) T_1^{-1}(Z) \log(T_1^{-1}(Z))\right]}
\equiv
I_{T_1},
\end{align}
with
\begin{align}\label{def-U}
	U_{T_1}(z)
	\equiv
		\left[{T_1''(T_1^{-1}(z))\over T_1'(T_1^{-1}(z))}-{T_1'(T_1^{-1}(z))\over z}\right] T_1^{-1}(z), \quad z>0.
\end{align}

Furthermore, since  $\{Y_i : i = 1,\ldots , n\}$ is a random sample from $Y$, the random vectors 
\begin{align*}
	\begin{pmatrix}
	T_1(Y_1)
	\\[0,15cm]
	\log(Y_1)
	\\[0,15cm]
	\left[{T_1''(Y_1)\over T_1'(Y_1)}-{T_1'(Y_1)\over T_1(Y_1)}\right] Y_1 \log(Y_1)
	\\[0,15cm]
	T_1'(Y_1) Y_1 \log(Y_1)
	\\[0,15cm]
	{T_1'(Y_1)\over T_1(Y_1)}\,  Y_1 \log(Y_1)
	\end{pmatrix},
	\ldots,
		\begin{pmatrix}
	T_1(Y_n)
	\\[0,15cm]
	\log(Y_n)
	\\[0,15cm]
	\left[{T_1''(Y_n)\over T_1'(Y_n)}-{T_1'(Y_n)\over T_1(Y_n)}\right] Y_n \log(Y_n)
	\\[0,15cm]
	T_1'(Y_n) Y_n \log(Y_n)
	\\[0,15cm]
	{T_1'(Y_n)\over T_1(Y_n)}\,  Y_n \log(Y_n)
	\end{pmatrix}
\end{align*}
 are independent and identically distributed with 
 \begin{align*}
 \boldsymbol{Y}=
 \begin{pmatrix}
 T_1(Y)
 \\[0,15cm]
 \log(Y)
 \\[0,15cm]
 \left[{T_1''(Y)\over T_1'(Y)}-{T_1'(Y)\over T_1(Y)}\right] Y \log(Y)
 \\[0,15cm]
 T_1'(Y) Y \log(Y)
 \\[0,15cm]
 {T_1'(Y)\over T_1(Y)}\,  Y \log(Y)
 \end{pmatrix}.
 \end{align*}
 {\color{black} Then, provided  that the components of the random vector $\boldsymbol{Y}$ have finite second moments, the multivariate central limit theorem (CLT) yields}
\begin{align*}
\sqrt{n}
\left[
{1\over n}\sum_{i=1}^{n}
\begin{pmatrix}
T_1(Y_i)
\\[0,15cm]
\log(Y_i)
\\[0,15cm]
\left[{T_1''(Y_i)\over T_1'(Y_i)}-{T_1'(Y_i)\over T_1(Y_i)}\right] Y_i \log(Y_i)
\\[0,15cm]
T_1'(Y_i) Y_i \log(Y_i)
\\[0,15cm]
{T_1'(Y_i)\over T_1(Y_i)}\,  Y_i \log(Y_i)
\end{pmatrix}
-
 \begin{pmatrix}
\mathbb{E}[T_1(Y)]
\\[0,15cm]
\mathbb{E}[\log(Y)]
\\[0,15cm]
\mathbb{E}\left\{\left[{T_1''(Y)\over T_1'(Y)}-{T_1'(Y)\over T_1(Y)}\right] Y \log(Y)\right\}
\\[0,15cm]
\mathbb{E}[T_1'(Y) Y \log(Y)]
\\[0,15cm]
\mathbb{E}\left[{T_1'(Y)\over T_1(Y)}\,  Y \log(Y)\right]
\end{pmatrix}
\right]
\stackrel{\mathscr{D}}{\longrightarrow}
N_5
\left(
\begin{pmatrix}
0
\\[0,15cm]
0
\\[0,15cm]
0
\\[0,15cm]
0
\\[0,15cm]
0
\end{pmatrix}
, \boldsymbol{\Sigma}
\right),
\end{align*}
where ``$\stackrel{\mathscr{D}}{\longrightarrow}$'' means convergence in distribution. So, delta method provides 
\begin{align}\label{partial-CLT}
\sqrt{n}
\left[
\begin{pmatrix}
\dfrac{{1\over n}\sum_{i=1}^{n} T_1(Y_i) 
	\left\{
	1+ 
	{1\over n}
	\sum_{i=1}^{n}
	\log(Y_i)
	+ 
	{1\over n}
	\sum_{i=1}^{n}
	\left[{T_1''(Y_i)\over T_1'(Y_i)}-{T_1'(Y_i)\over T_1(Y_i)}\right] Y_i \log(Y_i)
	\right\}}{
	{1\over n} \sum_{i=1}^{n} T_1'(Y_i) Y_i \log(Y_i)
	-
	{1\over n}\sum_{i=1}^{n} T_1(Y_i) \ {1\over n}\sum_{i=1}^{n} {T_1'(Y_i)\over T_1(Y_i)}\,  Y_i \log(Y_i)}
\\[1cm]
\dfrac{1}{{1\over n} \sum_{i=1}^{n} T_1(Y_i)}
\end{pmatrix}
-
\begin{pmatrix}
I_{T_1}
\\[0,1cm]
\sigma
\end{pmatrix}
\right]
\stackrel{\mathscr{D}}{\longrightarrow}
N_2
\left(
\begin{pmatrix}
0
\\[0,15cm]
0
\end{pmatrix}
, \boldsymbol{A}\boldsymbol{\Sigma}\boldsymbol{A}^\top
\right),
\end{align}
with $I_{T_1}$ being as in \eqref{id-3} and $\boldsymbol{A}$ being the partial derivatives matrix defined as
\begin{align}\label{def-A}
\boldsymbol{A}
=
\begin{pmatrix}
{\partial h_1(\boldsymbol{y})\over \partial y_1} &
{\partial h_1(\boldsymbol{y})\over \partial y_2} &
{\partial h_1(\boldsymbol{y})\over \partial y_3} &
{\partial h_1(\boldsymbol{y})\over \partial y_4} &
{\partial h_1(\boldsymbol{y})\over \partial y_5} 
\\[0,2cm]
{\partial h_2(\boldsymbol{y})\over \partial y_1} &
{\partial h_2(\boldsymbol{y})\over \partial y_2} &
{\partial h_2(\boldsymbol{y})\over \partial y_3} &
{\partial h_2(\boldsymbol{y})\over \partial y_4} &
{\partial h_2(\boldsymbol{y})\over \partial y_5} 
\end{pmatrix}
\Bigg\vert_{\boldsymbol{y}=\mathbb{E}(\boldsymbol{Y})},
\end{align}
and 
\begin{align*}
h_1(\boldsymbol{y})
=
\dfrac{y_1 (1+y_2+y_3)}{y_4-y_1 y_5},
\quad
h_2(\boldsymbol{y})
=
{1\over y_1},
\quad \boldsymbol{y}=(y_1,y_2,y_3,y_4,y_5).
\end{align*}
For simplicity of presentation, we do not present the partial derivatives of $g_j, \ j = 1, 2$, here.
From \eqref{mle-sigma-1} and \eqref{mle-mu-1}, the convergence in \eqref{partial-CLT} is rewritten as
\begin{align}\label{partial-CLT-1}
	\sqrt{n}\left[
	\begin{pmatrix}
	\widehat{\mu}
	\\[0,15cm]
	\widehat{\sigma}
	\end{pmatrix}
	-
	\begin{pmatrix}
	I_{T_1}
	\\[0,1cm]
	\sigma
	\end{pmatrix}
	\right]
	\stackrel{\mathscr{D}}{\longrightarrow}
	N_2
	\left(
	\begin{pmatrix}
	0
	\\[0,15cm]
	0
	\end{pmatrix}
	, \boldsymbol{A}\boldsymbol{\Sigma}\boldsymbol{A}^\top
	\right).
\end{align}

{\color{black}  Although explicit expressions for the asymptotic covariance matrix are generally difficult to obtain for arbitrary $T_1$, it can be estimated numerically. Alternatively, bootstrap methods may be used to approximate the sampling distribution of the proposed estimators.}

\begin{proposition}\label{prop-inq}
	If 
$
	U_{T_1}(z)=-1, \ z>0,
$
	then 
	\begin{align*}
	I_{T_1}=\mu.
	\end{align*}
	Furthermore, from \eqref{id-3},
	the estimator $\widehat{\mu}$ given in \eqref{mle-mu-1}  is  strongly consistent for $\mu$, and from \eqref{partial-CLT-1},
	\begin{align*}
	\sqrt{n}\left[
	\begin{pmatrix}
	\widehat{\mu}
	\\[0,15cm]
	\widehat{\sigma}
	\end{pmatrix}
	-
	\begin{pmatrix}
	\mu
	\\[0,1cm]
	\sigma
	\end{pmatrix}
	\right]
	\stackrel{\mathscr{D}}{\longrightarrow}
	N_2
	\left(
	\begin{pmatrix}
	0
	\\[0,15cm]
	0
	\end{pmatrix}
	, \boldsymbol{A}\boldsymbol{\Sigma}\boldsymbol{A}^\top
	\right),
	\end{align*}
	where $\boldsymbol{A}$ is as in \eqref{def-A} and $\boldsymbol{\Sigma}$ is the covariance matrix  of $\boldsymbol{Y}$.
\end{proposition}
\begin{proof}
Note that $U_{T_1}(z)=-1$ is a second-order nonlinear ordinary differential equation which can be written as
\begin{align*}
{{\rm d}^2 \log(T_1(x))\over{\rm d}x^2}
+
{1\over x} \,
{{\rm d} \log(T_1(x))\over{\rm d}x}
=
0, 
\quad 
\text{with} \ x=T_1^{-1}(z).
\end{align*}
%
%
It is not difficult to verify that the differential equation solution is given by
$T_1(x)=C x^{-s}$, for some $C>0$ and $s\neq 0$. Consequently,
$T_1^{-1}(x)=(x/C)^{-1/s}$,
$T_1'(x)=-Cs x^{-s-1}$,
$T_1''(x)=Cs(s+1) x^{-s-2}$,
$T_1''(T_1^{-1}(x))=Cs(s+1) (x/C)^{1+2/s}$ 
and
$T_1'(T_1^{-1}(x))=-Cs (x/C)^{1+1/s}$.
Hence, $I_{T_1}$ in \eqref{id-3} is written as
%
	\begin{align}
	I_{T_1}
	=
	\dfrac{1}{\displaystyle
		\mathbb{E}[(\sigma Z-1)\log(Z)]},
	\quad 
	Z\sim {\rm Gamma}\left(\mu, {\color{black} \mu\sigma}\right).
	\label{id-mu-conv}
	\end{align}
	Substituting the known identities 
	$\mathbb{E}[Z\log(Z)]=[\psi^{(0)}(\mu+1)-\log(\mu\sigma)]/\sigma$ and $\mathbb{E}[\log(Z)]=\psi^{(0)}(\mu)-\log(\mu\sigma)$, into \eqref{id-mu-conv}, it follows that
	$
	I_{T_1}(Z)= \mu.
	$
	Therefore, the proof has been completed.
\end{proof}

{\color{black} The condition $U_{T_1}(z)= -1$ in Proposition \ref{prop-inq} identifies a special subclass of transformation functions for which the asymptotic limit in \eqref{id-3} simplifies substantially. In particular, it is satisfied by the family $T_1(x)=Cx^{-s}$, which includes several classical distributions listed in Table \ref{table:1}.
}

In the following results we use the assertion that $f(x)=o(g(x))$ as $x\to \infty$  (read $f(x)$ is little-$o$ of $g(x)$). This
assertion means that $g(x)$ grows much faster than $f(x)$, and is equivalent to
\begin{align*}
	\lim_{x\to \infty} {f(x)\over g(x)}=0.
\end{align*}
{\color{black}
\begin{lemma}\label{formula-IZ}
	If $g_1(1/\sigma)\neq 0$,
	then the following holds:
	\begin{align*}
		I_{T_1}
		=
		{	
			\mu 
			+
			{g_1''({1/\sigma})\over 2\sigma^2 g_1({1/\sigma})}
			+
			o\bigl({{\sigma^{-2}}\over g_1({1/\sigma})}\bigr)
			\over 
			1
			+
			o\bigl({{\sigma^{-2}}\over g_1({1/\sigma})}\bigr)
		},
		\quad \sigma\to\infty,
	\end{align*}
	where
	\begin{align*}
		\begin{array}{lll}
			&g_1(z)
			=
			1+[1+U_{T_1}(z)]\log(T_1^{-1}(z))
			\\[0,4cm]
			&g_2(z)
			=
			\left(\sigma-\dfrac{1}{z}\right) T_1'(T_1^{-1}(z)) T_1^{-1}(z) \log(T_1^{-1}(z))
		\end{array}
		, \quad  z>0,
	\end{align*}
	and $U_{T_1}$ is defined in \eqref{def-U} and $I_{T_1}$ is given by
	$
	I_{T_1}
	=
	{\mathbb E[g_1(Z)]}/{\mathbb E[g_2(Z)]},
	$
	see \eqref{id-3}.
\end{lemma}
\begin{proof}
	Since $Z\sim {\rm Gamma}(\mu,\mu\sigma)$, we have
	\[
	\mathbb E(Z)=\frac1\sigma,
	\quad
	{\rm Var}(Z)=\frac1{\mu\sigma^2}\to0,
	\quad \sigma\to\infty.
	\]
	Moreover, if $Y=\mu\sigma Z\sim{\rm Gamma}(\mu,1)$, then Proposition~\ref{prop-ant} in the Appendix yields
	\[
	\mathbb E\left|Z-\frac1\sigma\right|^{2+\delta}
	=
	\frac{1}{(\mu\sigma)^{2+\delta}} \, 
	\mathbb E|Y-\mu|^{2+\delta}
	=
	O\!\left(\sigma^{-2-\delta}\right).
	\]
	Since each $g_i$ is twice continuously differentiable in a neighborhood of $1/\sigma$, Lemma~\ref{lem:taylor-expectation} in the Appendix gives
	\begin{align}\label{eq-1n}
		\mathbb E[g_i(Z)]
		=
		g_i\!\left(\frac1\sigma\right)
		+
		\frac{g_i''(1/\sigma)}
		{2\mu\sigma^2}
		+
		o\!\left({\sigma^{-2}}\right),
		\quad i=1,2.
	\end{align}
	A laborious calculation shows that
	\begin{align*}
		g_2''\left({1\over\sigma}\right)
		&=
		2\sigma^2
		\left\{ 
		1+\left[
		1
		+
		U_{T_1}\left({1\over\sigma}\right)
		\right]
		\log\left(T_1^{-1}\left({1\over\sigma}\right)\right) 
		\right\}
		\\[0,2cm]
		&=
		2\sigma^2
		g_1\left({1\over\sigma}\right),
	\end{align*}
	where $U_{T_1}$ is as in \eqref{def-U}. Since $g_2\!\left(1/\sigma\right)=0$, it follows that
	\begin{align}\label{eq-2n}
		\mathbb E[g_2(Z)]
		=
		\frac{
			g_1\left({1/\sigma}\right)}
		{\mu}
		+
		o({\sigma^{-2}}).
	\end{align}
	The result follows by combining \eqref{eq-1n}, \eqref{eq-2n}, and the definition of $I_{T_1}$.
\end{proof}

{\color{black}
By combining \eqref{id-3} with Lemma \ref{formula-IZ}, we obtain the following theorem. Its main advantage is that it provides a unified framework for establishing the strong consistency of the proposed estimator $\widehat{\mu}$. Rather than studying each distribution in the exponential family \eqref{pdf-1} separately, it suffices to verify a few analytical conditions involving $T_1$ and the associated function $g_1$. As illustrated in the subsequent examples, these conditions hold for a broad range of generators in Table \ref{table:1}.
}
\begin{theorem}\label{theorem-main}
	Let $g_1$ be as in Lemma \ref{formula-IZ}. 
	Under the conditions
	\begin{align}\label{conditions-theo}
		o\left({{\sigma^{-2}}\over g_1({1/\sigma})}\right)=o(1),
		\quad 
		{g_1''\left({1\over\sigma}\right)}
		={o}\left(\sigma^2 g_1\left({1\over\sigma}\right)\right),
		\quad 
		\sigma\to\infty,
	\end{align}
	the estimator $\widehat{\mu}$ given in \eqref{mle-mu-1} is strongly consistent for $\mu$.
\end{theorem}
}

{\color{black}


%

\begin{example}
If $T_1(x)=\exp(x)-1$, then
$T_1^{-1}(x)=\log(x+1)$,
$T_1'(x)=T_1''(x)=\exp(x)$
and
$T_1''(T_1^{-1}(x))=T_1'(T_1^{-1}(x))=x+1$.
Hence, $g_1$, defined in Lemma \ref{formula-IZ}, is written as
\begin{align*}
g_1(z)
=
1+\left[1-{\log(z+1)\over z}\right]\log(\log(z+1)).
\end{align*}
Hence
\begin{align*}
	g_1''(z)
	&=
	{2 \big[{\log(z + 1)\over z^2} - {1\over z(z + 1)}\big]\over (z + 1)\log(z + 1)}
	-
	\left[
	{2 \log(z + 1)\over z^3} - {2\over z^2 (z + 1)} - {1\over z (z + 1)^2}
	\right] 
	\log(\log(z + 1)) 
	\\[0,2cm]
	&- 
	\left[{1\over (z + 1)^2 \log^2(z + 1)} + {1\over (z + 1)^2 \log(z + 1)}\right] \left[1 - {\log(z + 1)\over z} \right].
\end{align*}
Algebric calculus show that
\begin{align*}
	\lim_{\sigma\to \infty}{{\sigma^{-2}}\over g_1({1/\sigma})}=0, \quad 
	\lim_{\sigma\to \infty}
	{g_1''({1/\sigma})\over\sigma^2 g_1({1/\sigma})}=0.
\end{align*}
%
%
Hence, the conditions in \eqref{conditions-theo} are satisfied. Then, by applying Theorem \ref{theorem-main}, the estimator $\widehat{\mu}$ given in \eqref{mle-mu-1} is strongly consistent for $\mu$.
\end{example}

\begin{example}
	If $T_1(x)=\log(x+1)$, then
	$
		T_1^{-1}(x)
		=
		\exp(x)-1,
		$
		$
		T_1'(x)
		=
		{1}/({x+1}),
		$
		$
		T_1''(x)
		=
		-{1}/{(x+1)^2},
		$
		$
		T_1'(T_1^{-1}(x))
		=
		\exp({-x}),
		$
		$
		T_1''(T_1^{-1}(x))
		=
		-\exp({-2x}).$
Hence, $g_1$, as defined in Lemma~\ref{formula-IZ}, takes the form
	\begin{align*}
g_1(z)
=
1+
\exp(-z)
\left[
1-\frac{\exp(z)-1}{z}
\right]
\log(\exp(z)-1).
	\end{align*}
	Hence
	\begin{align*}
g_1''(z)
=
\frac{
	\exp(-z)
	\bigl[z^3+z^2+2z-2\exp(z)+2\bigr]
	\log(\exp(z)-1)
}{
	z^3
}
+
\frac{
	2\bigl[z^2+z-\exp(z)+1\bigr]
}{
	z^2[\exp(z)-1]
}
-
\frac{
	1-\exp(-z)-z\exp(-z)
}{
	z\,(\exp(z)-1)^2
}.
	\end{align*}
Algebraic manipulations show that
	\begin{align*}
		\lim_{\sigma\to \infty}{{\sigma^{-2}}\over g_1({1/\sigma})}=0, \quad 
		\lim_{\sigma\to \infty}
		{g_1''({1/\sigma})\over\sigma^2 g_1({1/\sigma})}=0.
	\end{align*}
Since the conditions in \eqref{conditions-theo} hold, Theorem \ref{theorem-main} implies that the estimator $\widehat{\mu}$ in \eqref{mle-mu-1} is strongly consistent for $\mu$.
\end{example}

\begin{example}
	If $T_1(x)=\exp(x-1/x)$, then
	\[
	\begin{aligned}
		T_1^{-1}(x)
		&=
		\frac{\log(x)+\sqrt{\log^2(x)+4}}{2},
		\quad 
		T_1'(x)
		=
		\exp\left(x-{1\over x}\right)
		\left(1+\frac1{x^2}\right),
		\\[2mm]
		T_1''(x)
		&=
		\exp\left(x-{1\over x}\right)
		\left(
		1+\frac{2}{x^2}
		-\frac{2}{x^3}
		+\frac{1}{x^4}
		\right),
		\quad 
		T_1'(T_1^{-1}(x))
		=
		x\left[
		1+
		\frac{4}
		{\bigl(\log(x)+\sqrt{\log^2(x)+4}\bigr)^2}
		\right],
		\\[2mm]
		T_1''(T_1^{-1}(x))
		&=
		x\left[
		1+
		\frac{8}{A^2}
		-\frac{16}{A^3}
		+\frac{16}{A^4}
		\right],
		\quad
		A=\log(x)+\sqrt{\log^2(x)+4}.
	\end{aligned}
	\]
	Hence, $g_1$, defined in Lemma \ref{formula-IZ}, becomes
	\[
	g_1(z)
	=
	1+
	\frac{
		\left(\dfrac{\log(z)+\sqrt{\log^2(z)+4}}{2}\right)^2-1
	}{
		\left(\dfrac{\log(z)+\sqrt{\log^2(z)+4}}{2}\right)^2+1
	} \, 
	\log\left(
	\frac{\log(z)+\sqrt{\log^2(z)+4}}{2}
	\right).
	\]
	Then
	\[
	g_1''(z)
	=
	\frac{
		-r^6+7r^4+9r^2+1
		+
		4r^2(1-3r^2)\log(r)
	}
	{(r^2+1)^3\,z^2\,s^2}
	-
	\frac{
		\bigl[r^4+4r^2\log(r)-1\bigr]
		\bigl[s^2+s+\log(z)\bigr]
	}
	{(r^2+1)^2z^2s^3},
	\]
	where
	$
	r=r(z)
	=
	[{\log(z)+\sqrt{\log^2(z)+4}}]/{2},
	\
	s=s(z)
	=
	\sqrt{\log^2(z)+4}.
	$
	After some algebra, it follows that
	\[
	\lim_{\sigma\to\infty}
	\frac{\sigma^{-2}}{g_1(1/\sigma)}
	=
	0,
	\quad
	\lim_{\sigma\to\infty}
	\frac{g_1''(1/\sigma)}
	{\sigma^2 g_1(1/\sigma)}
	=
	0.
	\]
	Hence, the conditions in \eqref{conditions-theo} are satisfied. Therefore, by Theorem \ref{theorem-main}, the estimator $\widehat{\mu}$ given in \eqref{mle-mu-1} is strongly consistent for $\mu$.

\end{example}

%
%
%
%
%

}

\section{Simulation study}\label{sec:simulation}
In this section, we carry out a Monte Carlo simulation study for evaluating the performance of the proposed estimators. For illustrative purposes, we only present the results for the Nakagami, gamma, new log-generalized gamma and scaled inverse chi-squared distributions. In the following, we present the proposed ML estimators for each of the aforementioned distributions. We then propose a bootstrap bias-reduced version of the proposed ML estimators, as they are biased \citep{RLR2016}.

\subsection{Nakagami distribution.}
	By considering the parameters $\mu=m$, $\sigma=1/\Omega$ and generator $T_1(x)=x^2$ of
	the Nakagami distribution, given  in Table \ref{table:1}, from \eqref{mle-sigma-1} and \eqref{mle-mu-1} the estimators for $\Omega$ and $m$ are given by
	\begin{align*}
	&\widehat{\Omega}
	=
	{1\over n}
	{\sum_{i=1}^{n} Y_i^{2}}
	\end{align*}
	and
\begin{align*}
	\widehat{m}
=
\dfrac{{1\over n}
	{\sum_{i=1}^{n} Y_i^{2}}}{2\left[
	{1\over n}\sum_{i=1}^{n} Y^{2}_i \log(Y_i)
	-
		{1\over n}
	{\sum_{i=1}^{n} Y_i^{2}} \ \
	{1\over n}\sum_{i=1}^{n} \log(Y_i)\right]},
	\end{align*}
	respectively.
A slightly modified version (by taking $p/2$ instead $p$) of estimators for $\Omega$ and $m$ have appeared in  \cite{RLR2016}.

\subsection{Gamma distribution.}
The following estimators for $\beta$ and $\alpha$ of the the gamma distribution are given by:
%
\begin{align*}
&{\widehat{\beta}}
=
{{1\over n}
	\sum_{i=1}^{n} Y_i  \log(Y_i)
	-
	{1\over n}
	\sum_{i=1}^{n} Y_i
	\
	{1\over n}
	\sum_{i=1}^{n} \log(Y_i)}
\end{align*}
and
\begin{align*}
&\widehat{\alpha}
=
\dfrac{
	{1\over n}  \sum_{i=1}^{n} Y_i}{
	{1\over n} \sum_{i=1}^{n} Y_i  \log(Y_i)
	-
	{1\over n}
	\sum_{i=1}^{n} Y_i
	\
	{1\over n}
	\sum_{i=1}^{n} \log(Y_i)},
\end{align*}
respectively.
	The estimators for $\beta$ and $\alpha$ have appeared in  \cite{YCh2016}.
%

%

\subsection{New log-generalized gamma distribution.}
	By considering the parameters $\mu= {\alpha}$, $\sigma= {1/(\alpha \beta)}$ and generator $T_1(x)=\exp(x)-1$ of new log-generalized gamma (with $\delta=1$), given  in Table \ref{table:1}, from \eqref{mle-sigma-1} and \eqref{mle-mu-1} the estimators for $\beta$ and $\alpha$ are given by
	\begin{align*}
	{\widehat{\beta}}
	=
	\dfrac
	{{1\over n} \sum_{i=1}^{n}\exp(Y_i) Y_i \log(Y_i)
		-
		{1\over n}\sum_{i=1}^{n} [\exp(Y_i)-1] \ {1\over n}\sum_{i=1}^{n} {\exp(Y_i)\over \exp(Y_i)-1}\,  Y_i \log(Y_i)}
	{
		1+
		{1\over n}
		\sum_{i=1}^{n}
		\log(Y_i)
		-
		{1\over n}
		\sum_{i=1}^{n}
		{1\over \exp(Y_i)-1}\, Y_i \log(Y_i)
	}
	\end{align*}
and
\begin{align*}
\widehat{\alpha}
&=
\dfrac{{1\over n}\sum_{i=1}^{n} [\exp(Y_i)-1]
	\left[
	1+
	{1\over n}
	\sum_{i=1}^{n}
	\log(Y_i)
	-
	{1\over n}
	\sum_{i=1}^{n}
	{1\over \exp(Y_i)-1}\, Y_i \log(Y_i)
	\right]}{{1\over n} \sum_{i=1}^{n}\exp(Y_i) Y_i \log(Y_i)
	-
	{1\over n}\sum_{i=1}^{n} [\exp(Y_i)-1] \ {1\over n}\sum_{i=1}^{n} {\exp(Y_i)\over \exp(Y_i)-1}\,  Y_i \log(Y_i)},
\end{align*}
	 respectively.


\subsection{Scaled inverse chi-squared distribution.}
	By considering the parameters $\mu=\nu/2$, $\sigma=\tau^2$ and generator $T_1(x)=1/x$ of the scaled inverse chi-squared distribution, given in Table \ref{table:1}, from \eqref{mle-sigma-1} and \eqref{mle-mu-1} the estimators for $\tau^2$ and $\nu$ are given by
	\begin{align*}
	&\widehat{\tau}^{\,2}
	=
	{1\over {1\over n} \sum_{i=1}^{n} Y_i^{-1}}
	\end{align*}
	and
\begin{align*}
	&
	\widehat{\nu}
	=
	2\
	\dfrac{{1\over n} \sum_{i=1}^{n} Y_i^{-1}}{
	{1\over n} \sum_{i=1}^{n} Y_i^{-1} \
		{1\over n} \sum_{i=1}^{n} \log(Y_i)- {1\over n} \sum_{i=1}^{n}Y^{-1}_i \log(Y_i)
	},
	\end{align*}
	 respectively.
%


\bigskip
{\color{black}
Now, we propose bootstrap bias-reduced ML estimators as
\begin{equation*}\label{asy:08}
\widehat{\theta}^{*} = 2\widehat{\theta} - \frac{1}{B}\sum_{b=1}^{B}\widehat{\theta}^{(b)},
\end{equation*}
where $\widehat{\theta}$ is the standard ML estimator, $B$ is the number of bootstrap replications, and $\widehat{\theta}^{(b)}$ is the closed-form estimator computed from the $b$-th bootstrap sample. Each bootstrap sample is obtained by drawing $n$ observations with replacement from the original data, following the standard nonparametric bootstrap procedure \citep{Efron1979,EfronTibshirani1993}. This bias-correction formula is based on the plug-in principle; see also \cite{DavisonHinkley1997}.}
To study the performance of the proposed bootstrap bias-reduced estimators, we computed the relative bias (RB) and root mean square error (RMSE), as
\begin{eqnarray*}
 \widehat{\textrm{RB}}(\widehat{\theta}^{*}) =  \left|\frac{\frac{1}{N} \sum_{i = 1}^{N} \widehat{\theta}^{*(i)} - \theta}{\theta}\right| ,  \quad
\widehat{\mathrm{RMSE}}(\widehat{\theta}^{*}) = {\sqrt{\frac{1}{N} \sum_{i = 1}^{N} (\widehat{\theta}^{*(i)} - \theta)^2}},
\end{eqnarray*}
where $\theta$ and $\widehat{\theta}^{*(i)}$ are the true parameter value and its $i$-th bootstrap bias-reduced estimate, and $N$ is the number of Monte Carlo replications. The sample sizes and true values of the parameters considered are: $n \in \{20,50,100,200,400,600\}$; $m\in \{1,2,5,10,15\}$, $\Omega=1$ (Nakagami); $\alpha\in \{0.5,1,2,4,6\}$, $\beta=1$ (gamma and new log-generalized gamma); and $\nu\in \{1,2,5,10,15\}$, $\tau^2=1$ (scaled inverse chi-squared). The number of Monte Carlo replications was $N=1,000$ and the number of bootstrap replications was $B=200$. The required numerical evaluations were implemented using the \verb+R+ software through the use of some packages already available at \url{http://cran.r-project.org} and some new computational routines.

{\color{black} In order to motivate the use of bootstrap bias reduction, Figure~\ref{fig_comparison} presents the empirical RB of the standard ML and bootstrap bias-reduced ML estimators for the new log-generalized gamma distribution. From this figure, it is clear that the standard ML estimators are noticeably biased, especially for small sample sizes, whereas the bootstrap bias-reduced estimators substantially correct this bias. The results for the other distributions are similar.}

{\color{blue}
\begin{figure}[htb!]
	\vspace{-0.25cm}
	\centering
	{\includegraphics[height=5.5cm,width=5.5cm]{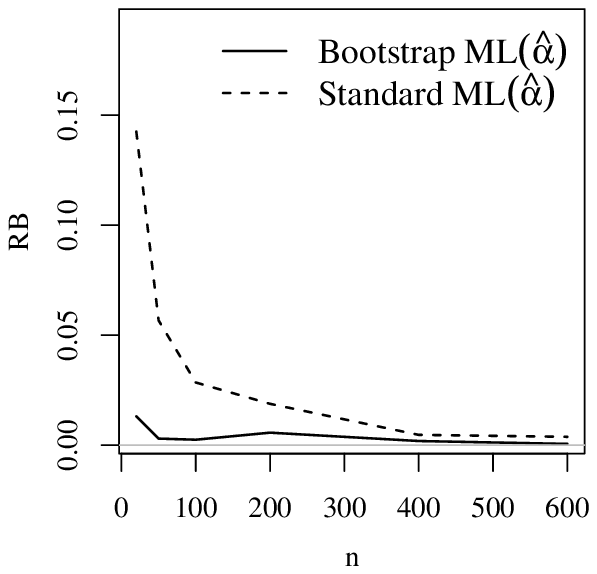}}\hspace{-0.25cm}
	{\includegraphics[height=5.5cm,width=5.5cm]{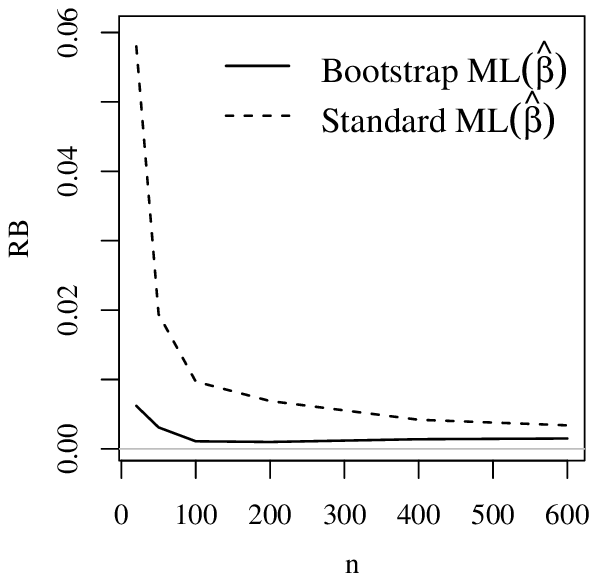}}\hspace{-0.25cm}
	\vspace{-0.2cm}
	\caption{{\color{black} Empirical RB of the standard ML and bootstrap bias-reduced ML estimators for the new log-generalized gamma distribution ($\alpha=0.5, \beta=1$).}}
	\label{fig_comparison}
\end{figure}
}

{\color{black} Figures~\ref{fig_dagum_MC1} and \ref{fig_dagum_MC2} show the results of the Monte Carlo simulation study for assessing the performance of the bootstrap bias-reduced ML estimators.} In terms of parameter recovery, from Figures~\ref{fig_dagum_MC1} and \ref{fig_dagum_MC2}, we observe that both RBs and RMSEs approach zero as the sample size ($n$) increases. In general terms, the results also show that the proposed bootstrap bias-reduced estimators have relatively low bias, even for small samples. Therefore, the proposed bootstrap bias-reduced ML estimators are good alternatives for estimating the parameters associated with the distributions in Table \ref{table:1}.

\begin{figure}[H]
\vspace{-0.25cm}
\centering
\subfigure[Nakagami]{\includegraphics[height=5.5cm,width=5.5cm]{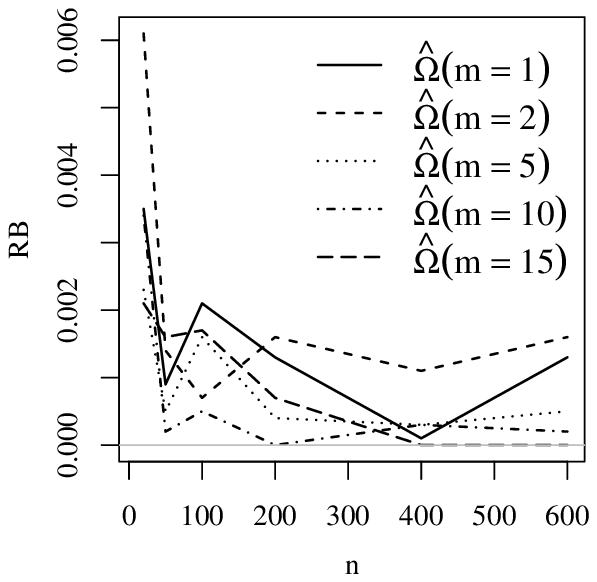}}\hspace{-0.25cm}
\subfigure[Nakagami]{\includegraphics[height=5.5cm,width=5.5cm]{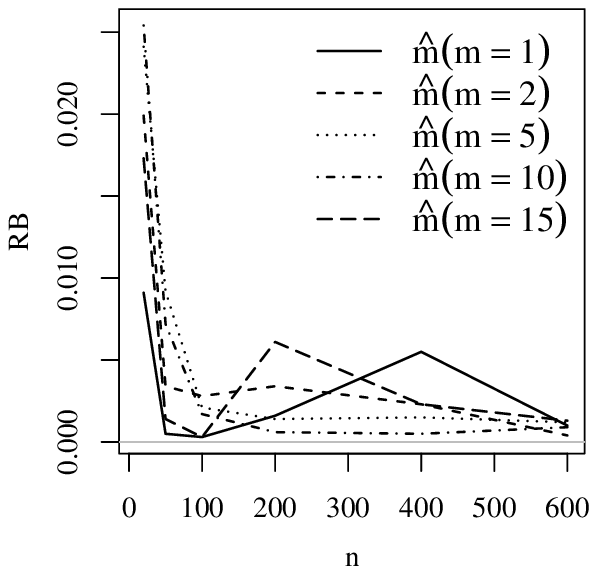}}\hspace{-0.25cm}
\subfigure[Gamma]{\includegraphics[height=5.5cm,width=5.5cm]{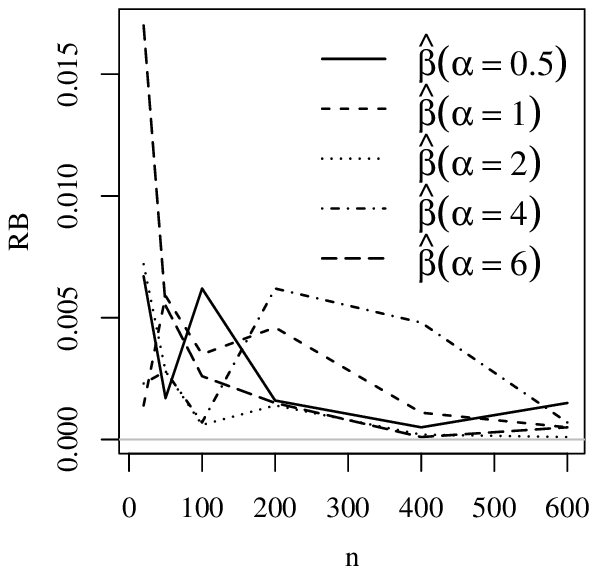}}\hspace{-0.25cm}
\subfigure[Gamma]{\includegraphics[height=5.5cm,width=5.5cm]{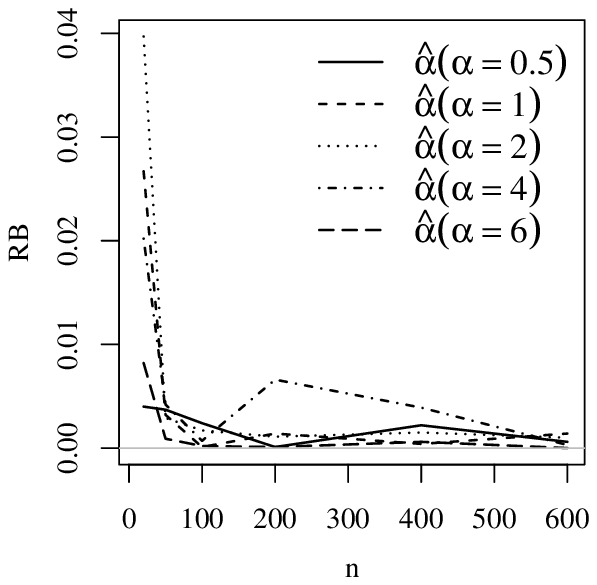}}\hspace{-0.25cm}
\subfigure[New log-generalized gamma]{\includegraphics[height=5.5cm,width=5.5cm]{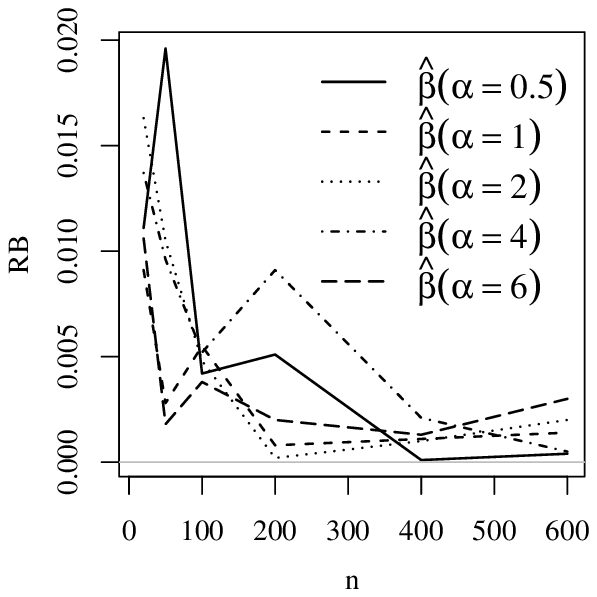}}\hspace{-0.25cm}
\subfigure[New log-generalized gamma]{\includegraphics[height=5.5cm,width=5.5cm]{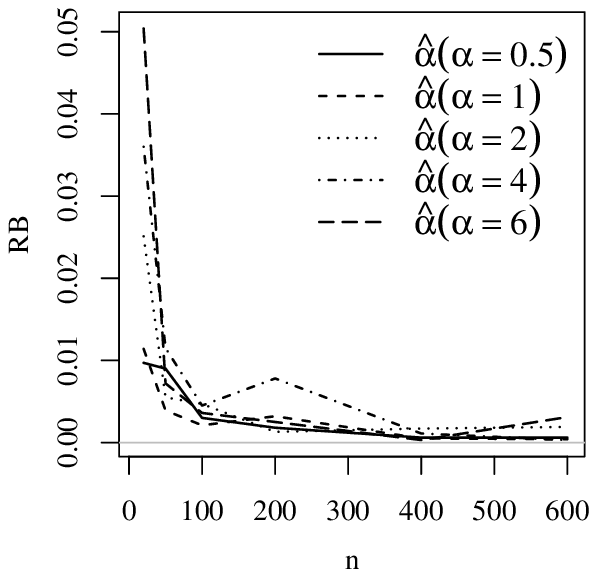}}\hspace{-0.25cm}
\subfigure[Scaled inverse chi-squared]{\includegraphics[height=5.5cm,width=5.5cm]{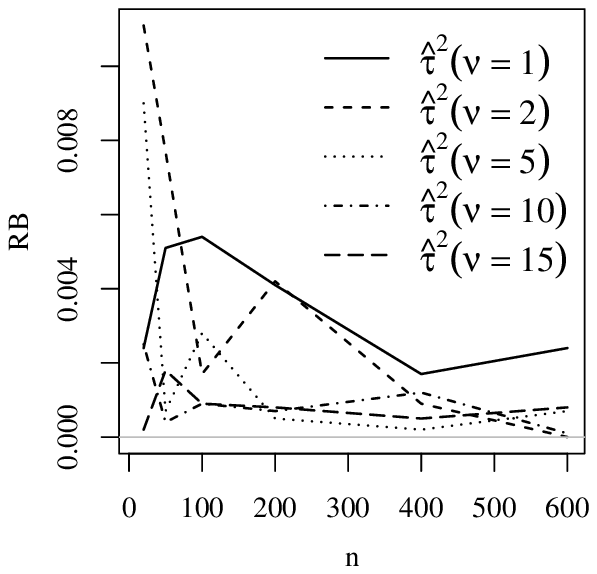}}\hspace{-0.25cm}
\subfigure[Scaled inverse chi-squared]{\includegraphics[height=5.5cm,width=5.5cm]{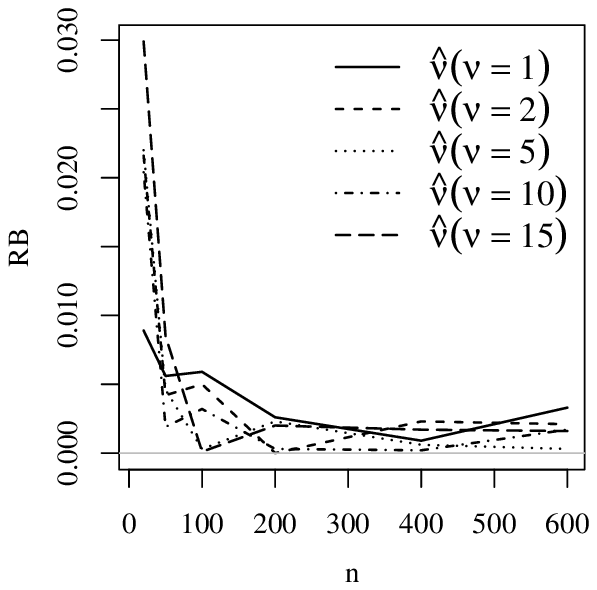}}\hspace{-0.25cm}
\vspace{-0.2cm}
\caption{Empirical RB of the bootstrap bias-reduced ML estimators for the {\color{black} indicated} distributions.}
\label{fig_dagum_MC1}
\end{figure}

\begin{figure}[H]
\vspace{-0.25cm}
\centering
\subfigure[Nakagami]{\includegraphics[height=5.5cm,width=5.5cm]{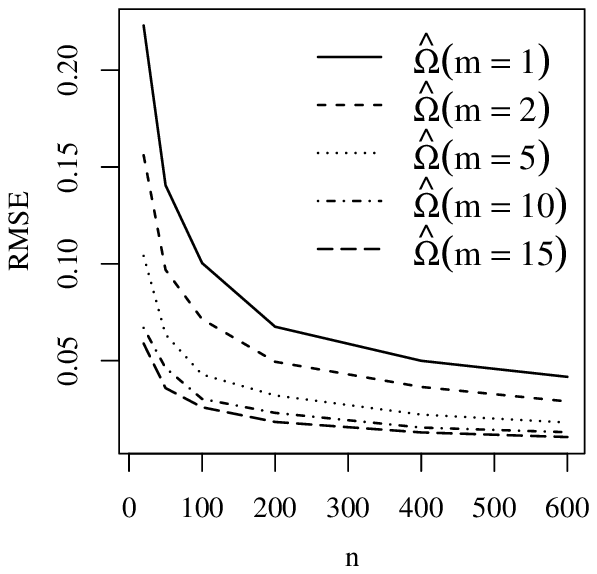}}\hspace{-0.25cm}
\subfigure[Nakagami]{\includegraphics[height=5.5cm,width=5.5cm]{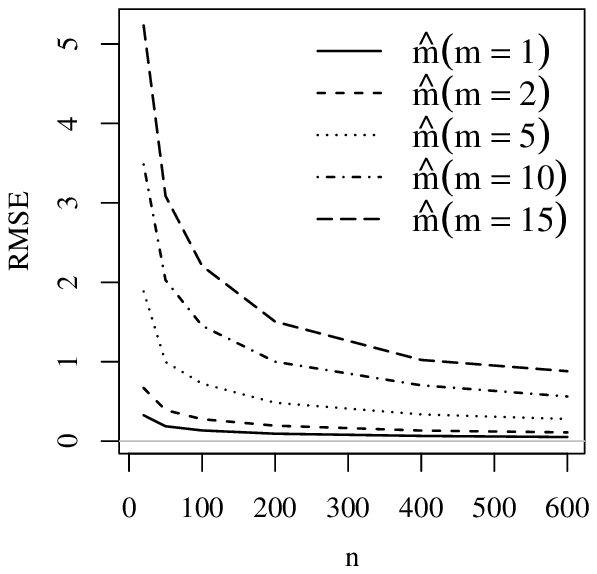}}\hspace{-0.25cm}
\subfigure[Gamma]{\includegraphics[height=5.5cm,width=5.5cm]{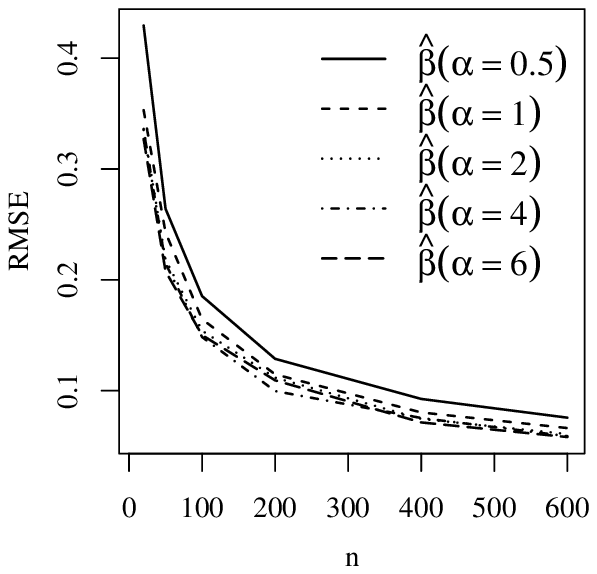}}\hspace{-0.25cm}
\subfigure[Gamma]{\includegraphics[height=5.5cm,width=5.5cm]{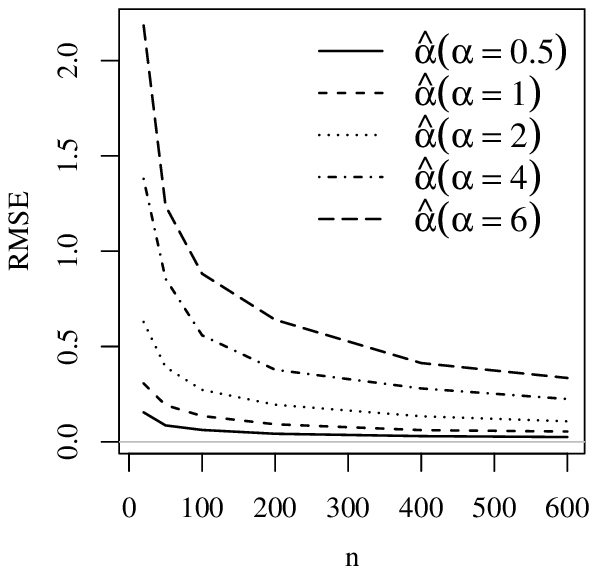}}\hspace{-0.25cm}
\subfigure[New log-generalized gamma]{\includegraphics[height=5.5cm,width=5.5cm]{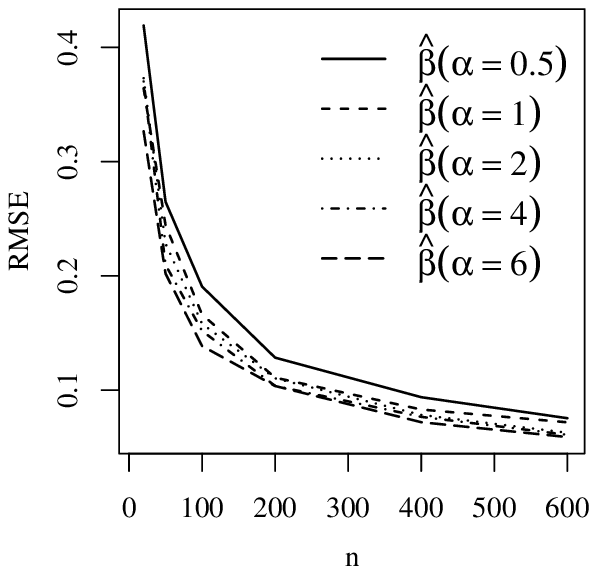}}\hspace{-0.25cm}
\subfigure[New log-generalized gamma]{\includegraphics[height=5.5cm,width=5.5cm]{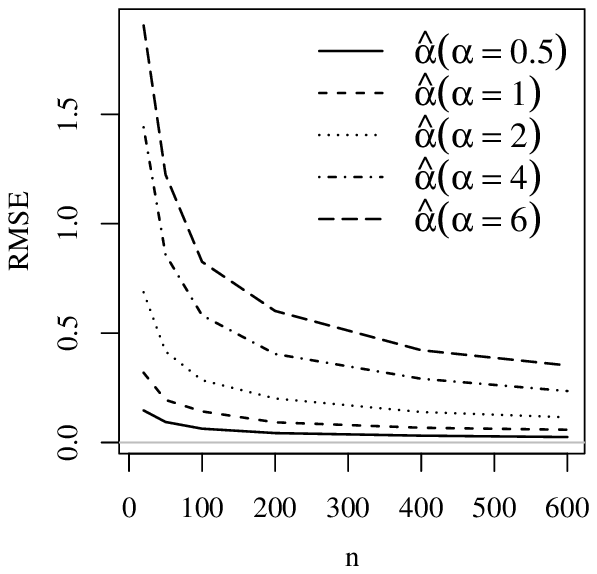}}\hspace{-0.25cm}
\subfigure[Scaled inverse chi-squared]{\includegraphics[height=5.5cm,width=5.5cm]{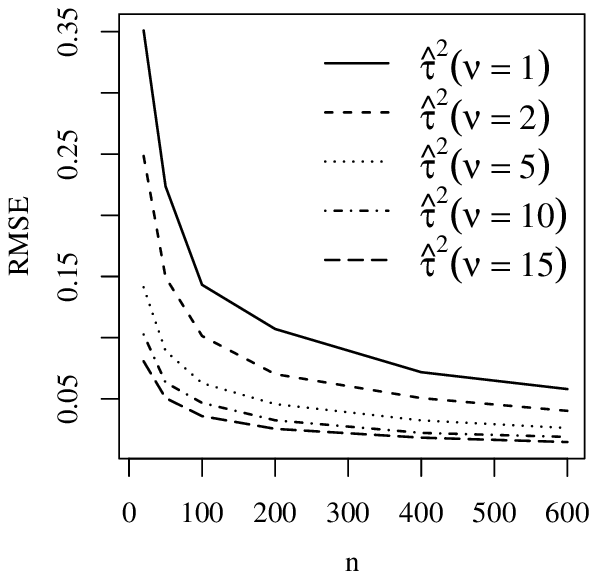}}\hspace{-0.25cm}
\subfigure[Scaled inverse chi-squared]{\includegraphics[height=5.5cm,width=5.5cm]{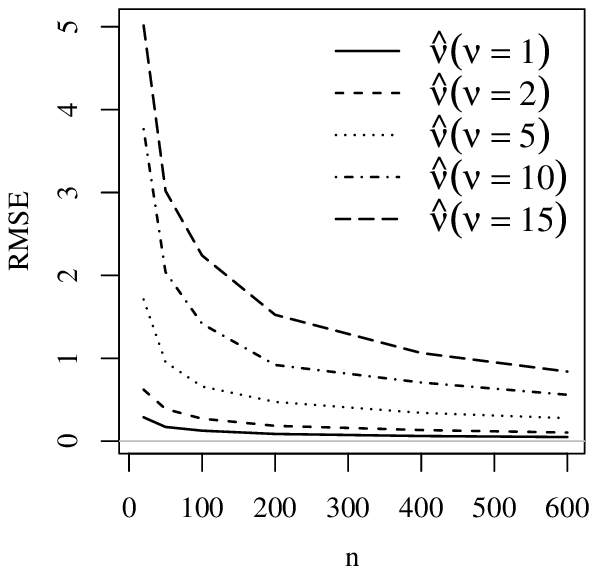}}\hspace{-0.25cm}
%
\vspace{-0.2cm}
\caption{Empirical RMSE of the bootstrap bias-reduced ML estimators for the {\color{black} indicated} distributions.}
\label{fig_dagum_MC2}
\end{figure}

%
%
%
%
%
%
%

	\paragraph{Acknowledgements}
 This study was financed in part by the
Coordenação de Aperfeiçoamento de Pessoal de Nível Superior - Brasil (CAPES) - Finance Code 001.

	\paragraph{Statements and Declarations}
	There are no conflicts of interest to disclose.




\begin{appendices}
\section{{\color{black} Some technical results} }\label{additional results}

{\color{black}
	In this section, we provide auxiliary technical results and asymptotic expansions for expectations of smooth functions of gamma random variables that facilitate the verification of the assumptions of Theorem \ref{theorem-main}.
	\begin{proposition}\label{prop-ant}
		Let
		$
		Y\sim{\rm Gamma}(\mu,1),
		$
		with $\mu>0$. Then, for every $r>0$ and every integer $m>r$, 
		\[
		\mathbb  E|Y-\mu|^r
		\leqslant 
		C_m^{\,r/m} \mu^{r/2},
		\]
		for some constant $C_m>0$ depending only on $m$.
	\end{proposition}
	
	\begin{proof}
		Let
		$
		W={(Y-\mu)}/{\sqrt{\mu}}
		$
		denote the standardized version of $Y$. 
		Since the cumulants of $W$ are
		$
		\kappa_1(W)=0,
		$
		$
		\kappa_2(W)=1,
		$
		and, for every $j\geqslant 3$,
		$
		\kappa_j(W)
		=
		(j-1)!\,\mu^{1-j/2},
		$
		we have
		\[
		\sup_{\mu\geqslant 1}
		|\kappa_j(W)|
		\leqslant 
		(j-1)!,
		\quad j\geqslant 2.
		\]
		
		Since every moment of order $m$ can be expressed as a polynomial in the
		first $m$ cumulants (see, e.g., the complete Bell polynomial formula),
		for every integer $m\geqslant 1$ there exists a constant $C_m>0$ such that
		\[
		\sup_{\mu\geqslant 1}
		\mathbb E|W|^m
		\leqslant C_m.
		\]
		Now let $r>0$ and choose an integer $m>r$. By Lyapunov's inequality,
		\[
		\mathbb E|W|^r
		\leqslant 
		(\mathbb E|W|^m)^{r/m}
		\leqslant
		C_m^{\,r/m}.
		\]
		
		Finally, since 	$
		W={(Y-\mu)}/{\sqrt{\mu}}
		$, the proof follows.
	\end{proof}

\begin{lemma}\label{lem:taylor-expectation}
	Let $\{X_n\}_{n\geqslant 1}$ be a sequence of random variables such that
	$
	\mathbb E[X_n]=\mu,
	$
	$
	\sigma_n^2={\rm Var}(X_n)\to0.
	$
	Assume that there exists $\delta>0$ such that
	$
	\mathbb E|X_n-\mu|^{2+\delta}
	=
	O\!\left(\sigma_n^{2+\delta}\right).
	$
	If $g$ is twice continuously differentiable in a neighborhood of $\mu$, then
	\[
	\mathbb E[g(X_n)]
	=
	g(\mu)
	+
	\frac{g''(\mu)}{2}\,\sigma_n^2
	+
	o(\sigma_n^2).
	\]
\end{lemma}

\begin{proof}
	%
	By Taylor's theorem,
	\[
	\mathbb E[g(X_n)]
	=
	g(\mu)
	+
	\frac{g''(\mu)}{2} \, \sigma_n^2
	+
	\mathbb E[R(X_n)],
	\]
	where
	\[
	R(x)
	=
	\frac12\bigl[g''(\mu+\theta(x)(x-\mu))-g''(\mu)\bigr](x-\mu)^2,
	\quad
	\theta(x)\in(0,1).
	\]
	
	It remains to prove that
	\begin{align}\label{rem-f}
	\mathbb E[R(X_n)]
	=
	o(\sigma_n^2).
		\end{align}
	
	Fix $\varepsilon>0$. Since $g''$ is continuous at $\mu$, there exists
	$\eta>0$ such that
	$
	|y-\mu|<\eta
	\ \Longrightarrow\ 
	|g''(y)-g''(\mu)|<\varepsilon.
	$
	
	Write
	$
	A_n=\{|X_n-\mu|\leqslant\eta\}.
	$
	Then
	\[
	|\mathbb E[R(X_n)]|
	\leqslant 
	\mathbb E\!\left[|R(X_n)|\,\mathds 1_{A_n}\right]
	+
	\mathbb E\!\left[|R(X_n)|\,\mathds 1_{A_n^c}\right].
	\]
	
	For the first term,
	\[
	\mathbb E\!\left[|R(X_n)|\,\mathds 1_{A_n}\right]
	\leqslant 
	\frac{\varepsilon}{2} \, 
	\mathbb E[(X_n-\mu)^2]
	=
	\frac{\varepsilon}{2}\,\sigma_n^2.
	\]
	
	Since $g''$ is bounded on a sufficiently small neighborhood of $\mu$,
	there exists $M>0$ such that
	$
	|g''(y)-g''(\mu)|\leqslant M
	$
	on that neighborhood. Therefore,
	\[
	\mathbb E\!\left[|R(X_n)|\,\mathds 1_{A_n^c}\right]
	\leqslant 
	\frac{M}{2} \, 
	\mathbb E\!\left[
	(X_n-\mu)^2
	\mathds 1_{\{|X_n-\mu|>\eta\}}
	\right].
	\]
	
	Using
	$
	(X_n-\mu)^2
	\mathds 1_{\{|X_n-\mu|>\eta\}}
	\leqslant
	\eta^{-\delta}|X_n-\mu|^{2+\delta},
	$
	we obtain
	\[
	\mathbb E\!\left[
	(X_n-\mu)^2
	\mathds 1_{\{|X_n-\mu|>\eta\}}
	\right]
	\leqslant 
	\eta^{-\delta}
	\mathbb E|X_n-\mu|^{2+\delta}
	\leqslant
	\eta^{-\delta}
	\{\mathbb E[(X_n-\mu)^2]\}^{1+\delta/2}
	=
	O\!\left(\sigma_n^{2+\delta}\right)
	=
	o(\sigma_n^2),
	\]
	where the last inequality follows from Lyapunov's inequality.
	
	Consequently,
	\[
	|\mathbb E[R(X_n)]|
	\leqslant
	\frac{\varepsilon}{2}\, \sigma_n^2
	+
	o(\sigma_n^2).
	\]
	
	Dividing by $\sigma_n^2$ and taking the upper limit,
	\[
	\limsup_{n\to\infty}
	\frac{|\mathbb E[R(X_n)]|}{\sigma_n^2}
	\leqslant 
	\frac{\varepsilon}{2}.
	\]
	
	Since $\varepsilon>0$ is arbitrary, \eqref{rem-f} follows. This completes the proof.
%
\end{proof}

}

\end{appendices}
\end{document}